\documentclass[12 pt]{article}

\usepackage[margin=1in]{geometry}  
\usepackage{graphicx}              
\usepackage{amsmath}               
\usepackage{amsfonts}              
\usepackage{dsfont}
\usepackage{amsthm}                
\usepackage{color}

\theoremstyle{plain}
\newtheorem{thm}{Theorem}[section]

\newtheorem{lem}[thm]{Lemma}
\newtheorem{prop}[thm]{Proposition}
\newtheorem{example}[thm]{Example}
\theoremstyle{definition}

\newtheorem{defn}{Definition}[section]
\theoremstyle{remark}
\newtheorem{rem}{Remark}[section]
\numberwithin{equation}{section}

\newcommand{\oh}{\overline{H}}
\newcommand{\lh}{\underline{H}}
\newcommand{\z}{\mathbb{Z}}
\newcommand{\com}{\mathbb{C}}
\newcommand{\q}{\mathbb{Q}}
\newcommand{\rr}{\mathbb{R}}
\newcommand{\n}{\mathbb{N}}
\newcommand{\h}{\mathbb{H}}

\newcommand{\bi}{\mathbf{i}}

\newcommand{\bo}{\mathbf{\Omega}}

\newcommand{\one}{\mathds{1}}
\newcommand{\ii}{\textbf{i}^{(n)}}
\newcommand{\ion}{\textbf{i}^{(n)}\in\Omega_2^{(n)}}

\newcommand{\zz}{\mathds{Z}}

\newcommand{\A}{\mathcal{A}}

\newcommand{\ket}[1]{\vert#1\rangle}

\newcommand{\jj}{\textbf{j}^{(n)}}

\begin{document}

\date{}

\title{Quantum entropy and complexity}

\author{F. Benatti$^{1,2}$, S. Khabbazi Oskouei$^3$, A. Shafiei Deh Abad$^4$\\[3mm]
\small\it $^1$Dipartimento di Fisica, Universit\`a di Trieste, I-34151 Trieste, Italy \\[1mm]  
\small\it $^2$Istituto Nazionale di Fisica Nucleare, Sezione di Trieste, I-34151 Trieste, Italy\\[1mm]
\small\it $^3$
Department of Mathematics, Islamic Azad University,\\[1mm]
\small\it Varamin-Pishva Branch, 33817-7489, Iran
\\[1mm]
\small\it $^4$Department of Mathematics, School of Mathematics, Statistics and Computer Science\\[1mm]
\small\it  College of Science, University of Tehran, Tehran, Iran
}

\maketitle

\begin{abstract}
We study the relations between the recently proposed machine-independent quantum complexity of P. Gacs~\cite{Gacs} and the
entropy of classical and quantum systems. On one hand, by restricting Gacs complexity to ergodic classical dynamical systems, we retrieve the equality between the Kolmogorov complexity rate and the Shannon entropy rate derived by
A.A. Brudno~\cite{Brudno}. On the other hand, using the quantum Shannon-Mc Millan theorem~\cite{BSM}, we show that such an equality holds densely in the case of ergodic quantum spin chains.
\end{abstract}

\section{Introduction}

The concept of algorithmic complexity introduced by
Kolmogorov, Solomonoff and Chaitin plays a fundamental role in connecting
the ergodic properties of classical dynamical systems to the predictability of
their trajectories~\cite{Brudno}.

Intuitively, any classical dynamical system can be encoded into a symbolic model by means of a suitable coarse-graining of the phase-space into a finite number of disjoint elementary cells and its trajectories can then be made correspond to sequences of symbols.
Consequently, portions of trajectories can ultimately be described by the shortest programs that, when run by universal Turing machines, provide as outputs the associated strings of symbols.
These shortest programs provide the best compressed descriptions of portions of trajectories, once symbolically encoded: the number of their symbols defines their algorithmic complexities with respect to the chosen coarse-graining. Dividing the number of symbols of the shortest description of a string by the number of symbols of the string itself and going with them to the limit of an infinitely long string, one finally associates an algorithmic complexity rate to any given trajectory with respect to the chosen coarse-graining. An absolute complexity rate can then be retrieved by taking the supremum over all possible finite coarse-grainings.

By its very definition, algorithmic complexity does not depend on a pre-assigned probability distribution over the ensemble of strings of a certain length. If a probability distribution is given, one associates an entropy rate to the statistical ensemble of sequences associated to the symbolic dynamical trajectories.  The largest entropy rate computed with respect to all possible finite coarse-grainings is known as Komogorov-Sinai dynamical entropy~\cite{CFS,billingsley} and a theorem of Brudno~\cite{Brudno} shows that, for classical ergodic systems, the algorithmic complexity rate equals the dynamical entropy rate for almost all trajectories with respect to the pre-assigned ergodic probability distribution.
Brudno's theorem can be seen as an extension to generic ergodic dynamics of Shannon-Mc Millan-Breiman theorem~\cite{BSM} which states that, for ergodic information sources, the Shannon entropy rate provides the maximal compression rate of the information source.

A natural question to ask is whether and how these notions and relations may be extended to a quantum setting; while the von Neumann entropy is the agreed upon quantum counterpart of the Shannon entropy, there are instead several proposals of quantum algorithmic
complexities~\cite{Berthiaume, Vitanyi, Mora, Gacs} and of quantum dynamical
entropies~\cite{CNT, Alicki, Sl, Do}.

Of the different quantum algorithmic complexities, we shall consider the one proposed by Gacs in~\cite{Gacs}: it extends to the quantum realm the notion of algorithmic probability thereby avoiding  reference to universal quantum Turing machines.

Purpose of this work is twofold: on one hand, we prove that, when restricted to ergodic classical dynamical systems, the Gacs complexity rate equals the Kolmogorov-Sinai dynamical entropy as by Brudno's theorem.
On the other hand, we show that, in the case of an ergodic quantum spin $1/2$ chain,  the Gacs complexity rate of typical one-dimensional projectors equals the von Neumann entropy rate of the chain.
Since, for sufficiently rapidly clustering states over the chain, the latter entropy rate equals the Connes-Narnhofer-Thirring (CNT) quantum dynamical entropy of the shift automorphism, the latter result is a non-commutative instance of the Brudno's relation.
With respect to a previous result along the same lines~\cite{FSA}, the present one is stronger in that, like in the classical formulation, it also specifies that the Gacs complexity rate is closed to the von Neumann entropy rate on a suitably dense set of pure state projections. This further information is achieved by means of a quantum generalization of the Shannon-McMillan theorem~\cite{BSM}.


\goodbreak


\section{Classical Dynamical Systems, Kolmogorov-Sinai Dynamical Entropy and Complexity}

In this section we briefly overview the basic tools concerning the dynamical entropy of classical dynamical systems, their algorithmic complexity and their relations.

\subsection{Kolmogorov-Sinai Dynamical Entropy}
\hfill\bigskip

\noindent
Classical, discrete-time, dynamical systems can be generically described by triplets  $(\mathcal{X}, T, \nu)$, where $\mathcal{X}$ is a measure space, called \textit{phase-space}, endowed with a $\Sigma$-algebra of measurable sets, $T$ is a \textit{measurable map} such that for any $A\in\Sigma$, $T^{-1}(A)\in\Sigma$ and $\nu$ is a \textit{$T$-invariant} probability measure on $\mathcal{X}$.

Any phase-space trajectory through a point $x\in\mathcal{X}$ then amounts to the collection $\{T^nx\}_{n\in\mathbb{Z}}$ of the images $T^nx$ of $x$ under the action of integer powers of the dynamical map $T$.

Finite measurable partitions $\mathcal{P}$ on $\mathcal{X}$ are finite collections of disjoint  measurable sets, \textit{atoms}, $P_1, P_2, \ldots, P_k$, such that
$\mathcal{X}=\cup_{i=1}^k P_i$.
They provide a \textit{coarse-graining} of the phase space with respect to which trajectories can be encoded by sequences $\mathbf{i}=\{i_j\}_{j\in\mathbb{Z}}$ of symbols $i_j\in\{1,2,\ldots,k\}$ such that $i_j$ labels the atom $P_{i_j}$ of $\mathcal{P}$ reached by $x$ at time $t=j$ so that $T^jx\in P_{i_j}$, or equivalently such that
$x\in T^{-j}(P_{i_j})$.

The set of strings of length $n$, $\ii=i_1i_2\ldots i_n$, will be denoted by $\Omega^{(n)}_k$, by $\Omega^*_k=\bigcup_{n=0}^\infty\Omega_k^{(n)}$ the set of strings of any finite length, and by $\Omega_k$ the set of all sequences of symbols from the alphabet $\{1,2,\ldots,k\}$.

In this way, by means of a finite measurable partition, to any classical dynamical system one associates a symbolic model $(\widetilde{\Omega}_k, \theta, \nu_\mathcal{P})$, where
\begin{enumerate}
\item
$\widetilde{\Omega}_k\subset\Omega_k$ is the subset of sequences corresponding to all trajectories of $(\mathcal{X}, T, \nu)$. This subset id endowed with the $\sigma$-algebra generated by
\textit{cylinder sets} consisting of all sequences whose elements have fixed
values in chosen intervals:
$$
C^{[j, k]}_{i_j i_{j+1}\ldots i_k}=\{\bi\in \Omega_k: \bi_{j+l}=\bi_{j+l}, \, l=0,1, \ldots, k-j\}\ ,
$$
where $\bi_\ell$ denotes the $\ell$-th entry of the string $\bi\in\widetilde{\Omega}^*_k$.
\item
$\theta$ is the shift dynamics along the sequences in $\Omega_k$:
$(\theta(\bi))_j=\bi_{j+1}$.
\item
The probability measure $\nu_{\mathcal{P}}$ is defined by the volumes of the cylinders
$C^{[0,n-1]}_{\ii}$:
\begin{eqnarray}
\nonumber
\nu_{\mathcal{P}}(\ii)&:=&\nu\Big(C^{[0,n-1]}_{\ii}\Big)\\
\label{refpart}
&=&\nu\Big(P_{i_0}\bigcap T^{-1}(P_{i_1})\bigcap \dots
\bigcap T^{-n+1}(P_{i_{n-1}}\Big)\ ,
\end{eqnarray}
where, given $i_j\in\{1,2,\ldots,k\}$,
$$
T^{-j}\Big(P_{i_j}\Big)=\Big\{\bi\in\Omega_k\,:\,(T^j\bi)_j=i_j\Big\}\ .
$$
\end{enumerate}

The Kolmogorov-Sinai entropy of $(\mathcal{X}, T, \nu)$ is given by the maximal Shannon entropy rate over all its symbolic models.
Namely, given a finite, measurable partition $\mathcal{P}$, its refinement
$\mathcal{P}^{(n)}$, namely the finite partition  whose atoms are the intersections in~\eqref{refpart}, has entropy
\begin{equation}
\label{Shannonent}
H_\nu(\mathcal{P}^{(n)}):=-
\sum_{\ii\in\Omega_k^{(n)}}\nu_{\mathcal{P}}(\ii)\log\nu_{\mathcal{P}}(\ii)\ .
\end{equation}
Because the probability measure $\nu$ is assumed to be $T$-invariant, when the string length goes to infinity, the Shannon entropy per symbol
$\displaystyle
\frac{1}{n}\,H_\nu(\mathcal{P}^{(n)})$ tends to the limit
\begin{equation}
\label{KSp}
h_\nu^{KS}(T, \mathcal{P}):=
\lim_{n\to\infty}\frac{1}{n}H_\nu(\mathcal{P}^{(n)})=\inf_n\frac{1}{n}H_\nu(\mathcal{P}^{(n)})\, .
\end{equation}
The Kolmogorov-Sinai (KS-) entropy of $(\mathcal{X}, T, \nu)$ is then defined by eliminating the dependence on the chosen partition:
\begin{equation}
\label{KS}
h_\nu^{KS}(T):=\sup_\mathcal{P}h_\nu^{KS}(T, \mathcal{P})\, ,
\end{equation}

\subsection{\bf Kolmogorov Complexity and Universal Probability}
\hfill
\bigskip

\noindent
We now briefly review the concepts of computability,
Kolmogorov complexity and universal semi-measures.

\begin{enumerate}
\item
A function from $\n^k$ to $\n$ is called partially computable if it
is computed by a Turing Machine ~\cite{Da} which, in the following, will be identified
with a suitable program or algorithm.
\item
There is a one-to-one correspondence $n\rightarrow P_n$ between the set $\n$ of natural numbers and the set of all programs with  fixed $k$ inputs~\cite{Da}. The
function computed by the program $P_n$ will be denoted by $\phi_n$.
\item
The mapping $n\rightarrow \phi_n$ provides a way of enumerating partially
computable functions. By a \textit{universal, partially computable
function} we mean a partially computable function
$\phi$ from $k+1$ integer inputs $(\{x_j\}_{j=1}^k, n)\in \n^k\times\n$ to $\n$ such that for each $\phi_n$,
$$
\phi(x_1, x_2, \ldots , x_k, n)=\phi_n(x_1, x_2, \ldots , x_k)\ .
$$
\item
Let $f:\n\times \n\to \rr$ be a function from integers to real numbers; for each $n\in\n$, let $f_n$ be defined by $f_n(x)=f(x,n)$. A function
$g:\n\to\rr$ is called \textit{lower semi-computable} if there
exists a computable function $f:\n\times\n\to\q$ into the rational numbers such that the
sequence $f_n$ is an increasing sequence and
$\lim_{n\rightarrow\infty}f_n=g$, for more details see \cite{FSA}.
A function $g:\n\to\rr$ is called \textit{upper semi-computable} if
$-g$ is lower semi-computable and \textit{computable} if it is lower and upper semi-computable.
\item
Given the set $\Omega^*_2$ of binary strings of any length, the map
\begin{equation}
\label{enc}
\Omega^*_2\ni\ii\mapsto\tau(\ii)=2^{n+1}- \sum_{k=1}^n i_k 2^{i_k}\ ,
\end{equation}
defines a one-to-one correspondence between $\Omega^*_2$ and $\n$.
\item
A function $f:\Omega^*_2\to \Omega^*_2$ is called \textit{partially
computable}, respectively \textit{computable} if the function
$\tau\circ f \circ \tau^{-1}:\n\to\n$ is partially computable,
respectively computable.
\item
Let $x$ and $y$ be two elements of $\Omega^*_2$ we say that $x$ is a
\textit{prefix} of $y$ if there is an element $z\in\Omega^*_2$ such
that $x.z=y$, where $x.z$ denotes the concatenation of $x$ and $z$.
A subset $S\subseteq\Omega^*_2$ is called \textit{prefix-free} if none
of its elements is a prefix of another element of $S$.
\item
A partially computable function on $\Omega^*_2$ is called a \textit{prefix-free
function} if its domain is a prefix-free subset of $\Omega^*_2$. In
the following by a \textit{(prefix)(universal) machine} we mean a
\textit{prefix-free, universal partially computable function}.
\end{enumerate}

\begin{defn}
\label{defAC}
A function $\mu:\n\to\rr$ is called a \textit{semi-computable,
semi-measure} if it is a positive semi-computable function such that
$\sum_x\mu(x)\leq 1$.
It has been shown (see for instance \cite{Vitanyi}) that there exists  a universal semi-computable  semi-measure $\mu$ in the sense that, for any
semi-computable semi-measure $\nu$, there exists a constant number
$c_\nu>0$ such that $c_\nu\,\nu(x)\leq\mu(x)$ for all $x\in\n$.
\end{defn}

The algorithmic complexity is a probability-independent measure of randomness of single
binary strings $\ii\in\Omega^*_2$ that hinges upon how difficult it is to describe them by algorithms, namely by binary programs $p$, that read by any universal Turing machine $\mathcal{U}$, reproduce those strings as outputs: $\mathcal{U}[p]=\ii$. This measure of complexity was introduced by Kolmogorov and Solomonoff~\cite{Solomonoff} and further elaborated by Chaitin~\cite{Chaitin}.

\begin{defn}
The \textit{algorithmic complexity}, or Kolmogorov complexity, of a
binary string $\ii\in \Omega^*_2$, is measured by the length, that is by the number of bits $\ell(p)=m^*$, of any shortest binary program $p$ such that
$\mathcal{U}[p]=\bi^{(n)}$:
\begin{equation}
\label{algcomp0} 
K(\bi^{(n)})=\min_p\Big\{\ell(p)\ :\
\mathcal{U}[p]=\bi^{(n)}\Big\}\ ,
\end{equation}
While Kolmogorov complexity is based on generic universal Turing machines,  Chaitin  complexity is instead based upon prefix-free, universal Turing machines:
\begin{equation}
\label{algcomp1} C(\bi^{(n)})=\min_p\Big\{\ell(p)\ :\
\mathcal{U}_{pf}[p]=\bi^{(n)}\Big\}\ ,
\end{equation}
where $\mathcal{U}_{pf}$ denotes a prefix-free, universal Turing machines.
\end{defn}
\medskip

\begin{rem}
\label{rem1}
~\begin{enumerate}
\item
Algorithmic complexity does not depend on the specific universal Turing machine $\mathcal{U}$.
Indeed, because of universality, any of them, $\mathcal{U}_1$, can reproduce the action of any other one, $\mathcal{U}_2$, so that the differences between the algorithmic complexities provided by them are related by an additive constant depending only on the translation program that makes $\mathcal{U}_1$ mimick $\mathcal{U}_2$, but not on the input string.
\item
Roughly speaking, because of the presence of patterns usable for compression, regular strings $\ii$ have complexities that scale as the $\log$ of the number of their bits, $m^*\simeq\log n$. On the contrary, complex strings are expected to be reproducible only by listing their bits, $m^*\simeq n$. For large $n\in\n$, most strings are expected to have large algorithmic complexity and be incompressible. Indeed, the number of binary programs with length smaller than $c>0$ cannot be larger than $2^c$ so that the cardinality of the set of strings $\bi$ with complexity $C(\bi)< c$ can be estimated by
\begin{equation}
\label{compression}
\# \{\bi | \bi\in\Omega^*_2 : C(\bi) < c   \} < 2^c\ .
\end{equation}
\item
Kolmogorov and prefix-free complexities are such that~\cite{Vitany}:
\begin{equation}
\label{re:kol}
K(\ii)\leq C(\ii)+2 \log{\ell(p)} + c_p \ ,
\end{equation}
where $c_p$ is a constant independent of $n$, while $\ell(p)$ serves to specify
where the program $p$ stops, an information which is necessary for the Kolmogorov complexity since, unlike for prefix-free universal Turing machines, another program $q$ can always be appended to the program $p$.
\item
As shown by ~\cite{Zvonkin70thecomplexity}, prefix-free complexity and universal
semi-measure are related by
\begin{equation}
\label{thm:levin2}
K(\bi)\overset{+}{=} -\log \mu(\bi)\ ,\quad \forall \bi\in\Omega^*_2\ ,
\end{equation}
where $f\overset{+}{=} g$ means that there exist constants $c_1$ and $c_2$
such that $f(\bi)\leq g(\bi)+ c_1$ and $g(\bi)\leq f(\bi)+c_2$, for
any string $\bi\in\Omega^*_2$.
\qed
\end{enumerate}
\end{rem}

\subsection{Brudno Relation}
\hfill
\bigskip

\noindent
As in the case of the entropy rate, we now consider the complexity rate,
namely, the complexity per symbol in the limit of larger strings:
\begin{equation}
\label{defci}
k(\bi):=\limsup_{n\to+\infty}\frac{1}{n}K(\ii)=\limsup_{n\to+\infty}\frac{1}{n}C(\ii)=:c(\bi)\ ,
\end{equation}
where $\ii$ is the initial prefix of length $n$ of the sequence in $\bi\in\Omega_k$ from a suitable alphabet $\{1,2,\ldots,k\}$.

Given a trajectory $\{T^jx\}_{j\in\z}$ through a point $x$ of a given phase-space $\mathcal{X}$ and a finite measurable partition, we shall denote by $\ii_{\mathcal{P}}(x)$ the string consisting of the labels of the atoms of $\mathcal{P}$ visited by $T^jx$ at times $0\leq j\leq n-1$. Then, the algorithmic complexity rate of a classical trajectory with respect to the coarse-graining of $\mathcal{X}$ given by $\mathcal{P}$ amounts to
\begin{equation}
\label{compratep}
c_{\mathcal{P}}(x):=\limsup_n\frac{1}{n}C(\ii_\mathcal{P}(x))\ .
\end{equation}
Finally, the partition-independent complexity of the trajectory through $x\in\mathcal{X}$ can then be defined by
\begin{equation}
\label{comprate}
c(x):=\sup_\mathcal{P} c_{\mathcal{P}}(x)\ .
\end{equation}
Notice that, because of \eqref{defci}, one can use the prefix algorithmic
complexity $K(\ii_{\mathcal{P}}(x))$ and obtain the rate $k(x)$ for all $x\in\mathcal{X}$ and $k(x)$ coincides with $c(x)$.

If $(\mathcal{X}, T, \nu)$ is  a reversible ergodic dynamical system, the
Shannon-Mc Millan-Breiman theorem connects the compressibility of a symbolic trajectory through $x\in\mathcal{X}$ to the KS-entropy \cite{billingsley}:
\begin{equation}
\label{thm:shannon-macmilla-breiman}
\lim_{n\to\infty} -\frac{1}{n}\log\nu(P_{\ii_{\mathcal{P}}(x)})=h^{KS}_\nu(T,
P)\quad \nu-a.e\ .
\end{equation}

Based on this, Brudno's theorem  \cite{Brudno} equates the algorithmic complexity rate for both prefix-free and non prefix-free  universal machines with the Kolmogorov-Sinai entropy for almost all $x\in\mathcal{X}$:
\begin{equation}
\label{fullBrudno}
c(x)=k(x)=h^{KS}_\nu(T,
P)\quad \nu-a.e\ .
\end{equation}

\begin{rem}
\label{rembrudno}
When the classical dynamical system is an ergodic information source,
$(\Omega_k,\theta,\pi)$, where $\theta$ is the left shift on $\Omega_k$, Brudno's theorem reduces to the assertion that (see \eqref{defci}):
\begin{equation}
\label{Brudnoeasy}
c(\bi)=k(\bi)=\limsup_{n\rightarrow\infty}\frac{1}{n}C(\textbf{i}^{(n)})=\limsup_{n\rightarrow\infty}\frac{1}{n}K(\textbf{i}^{(n)})=h^{KS}_\pi(\theta)\ ,
\end{equation}
for almost all sequences $\textbf{i}\in\Omega_k$ with respect to $\pi$.
\qed
\end{rem}

\section{Quantum Spin Chains}

Aim of the following sections is to extend the use of the classical notions and results introduced so far, in particular Brudno's relation,  to quantum systems.
Especially in view of the latest development in quantum information, communication and computation theories, these extensions may find applications in foundational issues.

As simple working instances of quantum systems that may comprise infinitely many degrees of freedom, we will focus upon quantum spin chains, namely upon one-dimensional lattices each site supporting a finite level quantum (spin) system that we shall fix to consist of two levels and thus to be described by a $2\times 2$ matrix algebra $M_2(\com)$.

Mathematically speaking, a quantum spin chain is the $C^*$-algebra that arises from the norm completion of local quantum spin algebras
\begin{equation}
\label{algtensor} M_{[-n,n]}=\underbrace{M_2(\com)\otimes
M_2(\com)\otimes\cdots M_2(\com)}_{2n+1\quad times} =M^{\otimes
2n+1}_2\ .
\end{equation}
In the norm-topology (the
norm is the one which coincides with the standard matrix-norm on
each local algebra) the limit $n\to+\infty$ of the nested sequence
$\{M_{[-n,n]}\}_{n\in\n}$ gives rise to the norm-complete infinite dimensional
\textit{quasi-local} algebra~\cite{BraRob}
\begin{equation}
\label{infalg} \mathcal{M}:=\lim_{n\to\infty} M_{[-n,n]}\ .
\end{equation}
Any local spin operator $A\in M_{[-n,n]}$ is naturally
embedded into $\mathcal{M}$ as
\begin{equation}
\label{embedding} M_{[-n,n]}\ni A\mapsto \one_{-n-1]}\otimes
A\otimes \one_{[n+1}\in\mathcal{M}\ ,
\end{equation}
where $\one_{-n-1]}$ stands for the infinite tensor products of
$2\times 2$ identity matrices up to site $-n-1$, while $\one_{[n+1}$
stands for the infinite tensor product of infinitely many identity
matrices from site $n+1$ onwards. In this way, the local algebras
are sub-algebras of the infinite one sharing  a same identity
operator.

The simplest dynamics on quantum spin chains is given by the
right shift $\Theta[M_{[-n,n]}]=M_{[-n+1,n+1]}$:
\begin{equation}
\label{shift}
\Theta[\one_{-n-1]}\otimes A\otimes \one_{[n+1}]=\one_{-n]}\otimes
A\otimes \one_{[n+2}\ .
\end{equation}
Any state $\omega$ on $\mathcal{M}$ is a positive, normalized linear
functional whose restrictions to the local sub-algebras are density
matrices $\rho_{[-n,n]}$, namely positive matrices in
$M_{[-n,n]}(\com)$ such that ${\rm Tr}_{[-n,n]}\rho_{[-n,n]}=1$:
\begin{equation}
\label{states} M_{[-n,n]}\ni A\mapsto \omega(A)={\rm
Tr}_{[-n,n]}\Big(\rho_{[-n,n]}\,A\Big)\ .
\end{equation}
The degree of mixedness of such density matrices is measured by the
von Neumann entropy
\begin{equation}
\label{vNent} S(\rho_{[-n,n]})=-{\rm
Tr}_{[-n,n]}(\rho_{[-n,n]}\log\rho_{[-n,n]})=-\sum_jr^{(n)}_j\log
r^{(n)}_j\ ,
\end{equation}
where $0\leq r^{(n)}_j\leq 1$, $\sum_jr^{(n)}_j=1$, are the
eigenvalues of $\rho_{[-n,n]}$. Notice that the von Neumann entropy
is nothing but the Shannon entropy of the spectrum of
$\rho_{[-n,n]}$ which indeed amounts to a discrete probability
distribution.

In the above expressions ${\rm Tr}_{[-n,n]}$ is the trace
computed with respect to any orthonormal basis of the Hilbert space
$\h_{[-n,n]}=(\com^2)^{\otimes 2n+1}$ onto which $A$ linearly acts.
Let $\vert i\rangle\in\com^2$, $i=0,1$, be such a basis in $\com^2$; then, a natural orthonormal basis in
$\h_{[-n,n]}$ will consist of tensor products of single spin
orthonormal vectors:
\begin{equation}
\label{ONBn} \vert\bi_{[-n,n]}\rangle=\bigotimes_{j=-n}^n\vert
i_j\rangle=\vert i_{-n}i_{-n+1}\cdots i_n\rangle\ ,
\end{equation}
namely its elements are indexed by binary strings
$\bi_{[-n,n]}\in\Omega_2^{(2n+1)}$. By going to the limit of an
infinite chain, a corresponding representation Hilbert space is
generated by orthonormal vectors, again denoted by
$\ket{\bi_{[-n,n]}}$, where $n$ arbitrarily varies and every
$\bi_{[-n,n]}$ is now a binary bi-infinite sequence in the set $\Omega_2^\z$ of such sequences where all
$i_k\notin[-n,n]$ are chosen equal to $0$. We shall denote by $\bi$
such binary strings, $\bo=(\Omega_2^\z)^*$ the set of binary strings of any
length and by $\ket{\bi}$ the
corresponding orthonormal vectors which form the so-called standard
basis of $\h$.

From~(\ref{states}), a compatibility relation immediately follows;
namely, for all $A\in M_{2^n}(\com)$, setting $\rho^{(n)}:=\rho_{[0,n-1]}$,
$$
\omega(A\otimes \one_{n})={\rm Tr}_{[0,n]}\Big(\rho^{(n+1)}\,A\otimes
\one_n\Big)=\omega(A)={\rm Tr}_{[0,n-1]}\Big(\rho^{(n)}\,A\Big)\ ,
$$
so that ${\rm Tr}_{n}\rho^{(n+1)}=\rho^{(n)}$.
On the other hand, if, for all $A\in M_{2^n}(\com)$\ ,
$$
\omega(\one_{0}\otimes A)={\rm Tr}_{[0,n]}\Big(\rho^{(n+1)}\one_{0}
\otimes A\Big)=\omega(A)={\rm Tr}_{[0,n-1]}\Big(\rho^{(n)}\,A\Big)\ ,
$$
that is if $\omega$ is a translationally invariant state, then
$\rho^{(n)}={\rm Tr}_{0}\rho^{(n+1)}_{[0,n]}$ for all $n\in\n$.

To any translationally invariant state $\omega$ on a quantum spin
chain there remains associated a well-defined von Neumann entropy
rate (see for instance \cite{Benattib}):
\begin{equation}
\label{vNentrate} s(\omega)=\lim_{n\to+\infty}\frac{1}{n}
S(\rho^{(n)})
\ .
\end{equation}

\section{Semi-computable semi-density matrices}

The concept of semi-computable semi-density matrices on infinite
dimensional separable Hilbert spaces is introduced in \cite{FSA}.
We will stick to the concrete case of quantum spin chains and start by noticing
that the vectors in~\eqref{ONBn} constitute an effectively constructed orthonormal basis
in the natural Hilbert space $\mathbb{H}$ associated with the chain.
\medskip

\begin{defn}
\label{algvect}
A vector state $\ket{\psi}=\sum_{\bi\in\bo}a_\bi|\bi>\in \mathbb{H}$, is
called elementary if only finitely many coefficients are non-zero algebraic numbers
and the remaining ones vanish.
A vector state $| \psi>=\sum_{\bi\in\bo}a_\bi|\bi>\in\h$ where $a_i\in\rr$,
will be termed semi-computable if there exist a computable sequence
of elementary vectors $| \psi_n> = \sum_{\bi\in \bo} a_{n,\bi}| \bi>$ and a
computable function $k: \n \rightarrow \q$, such that
$\lim_{n\rightarrow\infty} k_n = 0$, and for each $n, | a_\bi-
a_{n,\bi}| \leq k_n$.
\end{defn}

\begin{lem}\label{density matrix natural}
Elementary states are identified by natural numbers.
\end{lem}

\begin{proof}
A complex number $z$ is an algebraic number if it is a zero of the polynomial
$p(z)=x_0 z^{n}+x_1 z^{n-1}+\ldots+x_{n-1} z+x_n$ where $x_0,\ldots,x_n$ are (not all zero) integers.

The roots of any polynomial $p(z)=0$ can be arranged in the following lexicographical order: let $z_1=x_1+ i y_1$ and  $z_2=x_1+i y_2$; then, we say that $z_1$ preceeds $z_2$
if $x_1< x_2$ or if, when $x_1=x_2$, $y_1< y_2$. In this way, given the $n+1$-tuple of complex numbers $(z_0, \ldots z_n)$, one can introduce the quantities
$$
w(z_i)=2^n 3^{x'_0} \cdots p_{n+2}^{x'_n}p_{n+3}^i\ ,
$$
where $x'_j=f(x_j)$, with $f:\zz\to\n$ a one-to-one and surjective function, while $2,3,5,\ldots,p_{n+3}$ are prime numbers, listed in increasing order.

Let $\ket{\psi}=\sum_{\bi\in\bo}a_\bi|\bi>\in \mathbb{H}$ be an elementary state, with is algebraic numbers $a_{\bi}$. The enumeration of the elementary vectors can then be accomplished by associating them with the integer numbers
$$
W(\psi)= 2^n 3^{w(a_0)} \cdots p_{n+2}^{w(a_n)}\ ,
$$
where $n$ is the smallest  number such that  $a_{\bi}= 0$, for $\bi\notin [-n, n]$, when representing the bilateral sequence as an integer.
\end{proof}

\begin{example}
Consider the Bell state $\displaystyle |\psi>=\frac{|00>-|11>}{\sqrt{2}}$ with $\vert0\rangle$ and $\vert 1\rangle$, where $\vert 0\rangle$ and $\vert 1\rangle$ are the eigenvctors of $\sigma_3$ and constitute  an effectively constructible orthonormal basis in $\mathbb{C}^2$.
It is an elementary state as the coefficients $a_{00}=-a_{11}=1/\sqrt{2}$ are the roots of the equation $2z^2-1=0$.
As a function $f$ in the Lemma, one can use the following representation of rational numbers by integers,
$$
\frac{p}{q}\mapsto <0, <p, q>>\ ,\qquad -\frac{p}{q}\mapsto <1, <p, q>>\ ,
$$
where $p, q\geq 0$ and $<p, q>=2^p (2 q+1)-1$. It follows that the algebraic equation coefficients $x_0=2$ and $x_2=-1$ can be represented by
$$
x'_0=f(2)=<0, <2, 1>>=22\ ,\quad x'_2=f(-1)=<1, <1, 1>>=21\ ,
$$
whence
\begin{eqnarray*}
w\Big(\frac{1}{\sqrt{2}}\Big)&=&2\times 3^{22}\times 5^{21}\ ,\quad w\Big(-\frac{1}{\sqrt{2}}\Big)=2\times 3^{10}\times 5^{45}\times 7\\
W(\psi)&=&
2\times 3^{2\times 3^{22}\times 5^{22}}\times 5^{2\times 3^{22}\times 5^{21}\times 7}\ .
\end{eqnarray*}
\end{example}
\bigskip

\begin{defn}
\label{elemop}
A linear operator $T:\h_{[-n,n]}\rightarrow \h_{[-n,n]}$, will be called elementary if
the real and imaginary parts of all of its matrix entries are
rational numbers.
\end{defn}
\medskip

\begin{defn}
\label{semidm}
A linear operator $T:\h\rightarrow \h$, is a semi-density matrix
if  $T$ is positive and $0\leq {\rm Tr}(T)\leq 1$.
\end{defn}
\medskip

\begin{defn}
\label{quasi-inc-def}
Let  $n_1,n_2\in\n$ and $n_1\leq n_2$. Let
$T_j:\h_{[-n_j,n_j]}\rightarrow \h_{[-n_j,n_j]}$, $j=1,2$, be two
linear operators: $T_2$ will be said to be \textit{quasi-greater} than $T_1$,
$T_1\leq_q T_2$, if $P_{n_1}\,T_2\,P_{n_1}-T_1\geq0$, where
$P_n$ is the canonical projection from the infinite dimensional Hilbert space
$\mathbb{H}$ onto the finite dimensional subspace $\mathbb{H}_{[-n, n]}$, namely
$\mathbb{H}_{[-n, n]}=P_n\, \mathbb{H}\, P_n$.
A sequence of linear operators $T_n:\h_{[-n,n]}\rightarrow
\h_{[-n,n]}$ will be called quasi-increasing if for all $n\geq 1$,
$T_{n+1}\geq_q T_n$.
\end{defn}
\medskip

\begin{lem}
\label{lemaux}
Quasi-increasing sequences of positive operators $\{T_n\}_{n\in\n}$ of trace smaller than $1$ converge in trace norm to positive operators.
\end{lem}

\proof
As already shown in~\cite{FSA}, since the sequence $T_n$ is quasi-increasing, ${\rm Tr}(T_n)$ is an increasing sequence and since for every $n$, ${\rm Tr}(T_n)\leq 1$,
the sequence converges in trace-norm, $\|X\|_{tr}={\rm
Tr}\sqrt{X^\dag X}$ to an operator $T$ in the Banach space $T(\h)$
of trace-class operators on $\h$, moreover
$$
{\rm Tr}(T)=\lim_{n\to+\infty}{\rm
Tr}(T_n)=\lim_{n\to+\infty}\|T_n\|_{tr}=\|T\|_{tr}\leq 1\ .
$$
Therefore, $T$ must be positive.
\qed

\begin{rem}
\label{remaux}
{\rm Since the semi-density matrices $\rho^{(n)}_m$ and $\rho^{(n)}$ are compact positive operators,
their eigenvalues  $\{\lambda_k(\rho^{(n)}_m)\}_k$ and $\{\lambda_k(\rho^{(n)})\}_k$ can be listed in decreasing order.
Then,~\cite{E.Kowalski}
$$
\lim_{m\to\infty}\lambda_k(\rho^{(n)}_m) = \lambda_k(\rho^{(n)})\ .
$$
Therefore, the fact that eigenvalues and eigenvectors of elementary matrices are computable, the eigenvalues and eigenvectors of $\rho^{(n)}$ are semi-computable.
Moreover, each elementary density matrix $\rho^{(n)}_m$ is identified by a natural number $a_{m n}$ (see \eqref{density matrix natural}).}
\hfill\qed
\end{rem}

The previous result makes meaningful the following

\begin{defn}
\label{semicomp}
A linear operator $T$ on $\h$ is a
semi-computable semi-density matrix, if there exists a computable
\textit{quasi-increasing} sequence of elementary semi-density matrices
$\{T_n\}_n$ such that the trace-norm of their difference satisfies
$\lim_{n\rightarrow\infty}\|T-T_n\|_{tr}=0$.
\end{defn}
\medskip

Furthermore, in~\cite{FSA} it has been shown that there exists  a universal
semi-computable semi-density matrix $\hat{\mu}$, in the following
sense.

\begin{thm}
\label{univdem}
For any semi-computable semi-density matrix $\rho$ there exists
a constant number $c_\rho>0$ such that $c_\rho\,\rho\leq\hat{\mu}$.
\end{thm}
\medskip

The universal semi-computable semi-density matrix $\hat{\mu}$ can always be spectralized as
\begin{equation}
\label{usdm}
\hat{\mu}=\sum_{\bi\in\Omega_2}
\mu_{\bi}\vert\mu_{\bi}\rangle\langle\mu_{\bi}\vert\ ,
\end{equation}
with only a countable number of eigenvalues $\mu_\bi\neq 0$.

Using the notion of universal semi-density matrix, in analogy with the classical relation between the Kolmogorov complexity and the logarithm of a universal semi-computable semi-measure, one can introduce the following \textit{quantum complexity}, which we shall refer to as Gacs complexity in the following.
\medskip

\begin{defn}
\label{Gacscomp}
The Gacs complexity of a semi-computable semi-density matrix
$\rho\in B(\mathbb{H})$ is defined by
\begin{equation}
\label{gacsent}
\oh(\rho)=-{\rm Tr}(\rho \log\hat{\mu})\ .
\end{equation}
\end{defn}
\bigskip

\begin{rem}
\label{gacsrem}
In~\cite{Gacs} another possible quantum complexity was proposed,
$$
\lh(\rho)=-\log{\rm Tr}(\rho \hat{\mu})\ .
$$
Notice that $\oh(\rho)$ is more inherently quantum than $\lh(\rho)$ since $\oh(\rho)$ amounts to the mean value of the complexity operator $\kappa(\hat\mu):=-\log\hat\mu$.
Also, by the convexity of $-\log x$, it cannot be lower than $\lh(\rho)$~\cite{Gacs}.
\qed
\end{rem}

\section{Gacs Complexity: Classical dynamical Systems}

\begin{defn}
Let $(\mathcal{X}, T, \nu)$ be a dynamical system and $\mathcal{P}$ be a
finite measurable partition of $\mathcal{X}$. The associated symbolic model $(\Omega_k, T_\sigma, \nu_p)$  is called a \textit{semi-computable
symbolic model}  if $\nu_p$ as a function from
$\Omega_k$ into $\rr$ is a semi-computable probability measure.
\end{defn}
\medskip

\begin{rem}
\label{reminf}
Notice  that since $\sum_{\ii\in\Omega^{(n)}_k} \nu_\mathcal{P}(\ii)=1$, $\nu_\mathcal{P}$ is computable, but not a measure on $\Omega_k$; indeed,
$$
\sum_{\bi\in\Omega^*_k}\nu_\mathcal{P}(\bi)=\sum_{n=1}^\infty
1=\infty\ .
$$
\qed
\end{rem}

We will follow Gacs approach to quantum algorithmic complexity and adapt it to a classical framework in order to define the Gacs complexity rate for trajectories of classical dynamical systems. The purpose is to obtain in this way a proof of the classical Brudno's theorem by means of semi-computability techniques.

We first define the Gacs algorithmic complexity for a symbolic model
$(\Omega_k,\theta,\nu_\mathcal{P})$ of any dynamical system $(\mathcal{X},T,\nu)$, the symbolic model being constructed upon assigning a finite, measurable partition $\mathcal{P}$ with $k$ atoms.

In the classical case, the semi-computable semi-density matrices  are replaced by semi-computable semi-measures and universal semi-density matrices by universal semi-measures. Therefore, we can adapt the quantum definition and introduce the
notion of Gacs complexity of the semi-computable probability measure $\nu_{\mathcal{P}}$ as follows
\begin{equation}
G_{\nu_\mathcal{P}}(\mathcal{P}^{(n)}):=-\sum_{\ii\in\Omega_k^{(n)}}
\nu_\mathcal{P}(\ii)\log{\mu(\textbf{i}^{(n)})}\ ,
\end{equation}
where   $\mu$ is a universal semi-computable semi-measure  on
$\Omega_k$ and $\mathcal{P}^{(n)}$ is the partition with atoms the intersections
in~\eqref{refpart}.

The interpretation of the above definition is that $-\log\mu(\ii)$ represents the universal information content of the string $\ii$ so that its average with respect to
$\nu_{\mathcal{P}}$ measures the algorithmic randomness content of $\nu_{\mathcal{P}}$.

\begin{rem}
\label{remm}
Notice that, because of the universality of $\mu$, $\mu(\textbf{i}^{(n)})>0$  for all $\ii\in\Omega_k^{(n)}$. Indeed, let us define the function $f_{\ii}:\Omega_k\to \rr$ such that $f_{\ii}(\jj)=1$ for $\jj=\ii$, otherwise $0$; by construction, $f_{\ii}$ is semi-computable so that there exists a constant $c_f>0$ such that $c_f f_\ii(\jj) \leq \mu(\jj)$, for any $\jj\in\Omega_k$. Therefore, $\mu(\ii)>0$. Therefore, $G_\nu(\mathcal{P}^{(n)})$ is well defined.
\qed
\end{rem}

\begin{defn}
Let $(\mathcal{X}, T, \nu)$ be a  dynamical system.
Given a partition $\mathcal{P}$ such that $\nu_{\mathcal{P}}$ is  computable, one associates to it a Gacs complexity rate naturally given by
\begin{equation}
\label{eq:classical gacs}
g_\nu(T, \mathcal{P})=\limsup_{n\rightarrow\infty}\frac{1}{n}G_\nu(\mathcal{P}^{(n)})\ ,
\end{equation}
and also a Gacs complexity rate of $(\mathcal{X}, T, \nu)$ defined as
\begin{equation}
\label{defclassGacs}
g_\nu(T):= \sup_{\mathcal{P}} g_\nu(\mathcal{P})\ ,
\end{equation}
\end{defn}

\begin{rem}
Unlike for the Kolmogorov-Sinai entropy (see~\eqref{KS}), where the supremum is taken over all possible finite, measurable partitions, in order to be able to use semi-computability techniques, we need restrict the supremum in \eqref{defclassGacs} to be computed over finite measurable partitions $\mathcal{P}$ such that the corresponding measures $\nu_{\mathcal{P}}$ are semi-computable.
\qed
\end{rem}

Now, the natural question is whether there does exist a Brudno-like relation similar to \eqref{fullBrudno} between the Gacs algorithmic complexity rate and the KS-entropy in ergodic classical dynamical systems.
We start with an inequality involving the Gacs complexity rate and the Kolmogorov-Sinai entropy that is independent from ergodicity.

\begin{lem}
\label{thm:class}
Given a dynamical system $(\mathcal{X}, T, \nu)$,
\begin{equation}
g_\nu(T)\leq h_\nu^{KS}(T)\ .
\end{equation}
\end{lem}

\begin{proof}
Let $\mathcal{P}$ be a finite measurable  partition  of $\mathcal{X}$ such that
$\nu_\mathcal{P}$ is computable.
Since $\nu_\mathcal{P}$ is not a well defined measure on $\Omega_k$ (see Remark \ref{reminf}), we consider the following semi-computable measure  $f$ on
$\Omega_k$,
$$
f(\bi)=\frac{\delta(n)\nu_\mathcal{P}(\bi)}{\sum_{n=1}^\infty \delta(n)}\ , \quad \bi\in \Omega_k\ ,
$$
where $\delta(n)=\frac{1}{n \log^2{n}}$ so that $\sum_{\bi\in\Omega_k}f(\bi)=1$.
Then,
there exists a constant $c_{\nu_\mathcal{P}}>0$, dependent  on
$\nu_\mathcal{P}$, such that for any $\bi\in\Omega_k$,
$$
c_{\nu_\mathcal{P}}\delta(n)\nu_{\mathcal{P}}(\bi)\leq\mu(\bi)\ .
$$
Thus, restricting to strings of definite length $n$,
\begin{eqnarray*}
-\sum_{\ii\in\Omega^{(n)}_k} \nu_\mathcal{P}(\ii)\log{\mu(\textbf{i}^{(n)})}&\leq&
-\sum_{\ii\in\Omega^{(n)}_k}\nu_\mathcal{P}(\ii)\log\nu_\mathcal{P}(\ii)\\
&-&\log c_{\nu_{\mathcal{P}}}-\log\delta(n)\ .
\end{eqnarray*}
Then, using \eqref{KSp}, one gets
$$
g_\nu(T,\mathcal{P})\leq h_\nu^{KS}(T, \mathcal{P})\leq h_\nu^{KS}(T)\ ,
$$
so that taking the supremum over all partitions $\mathcal{P}$ such that
$\nu_{\mathcal{P}}$ is computable yields
$g_\nu(T)\leq h_\nu^{KS}(T)$.
\end{proof}

This Lemma allows us to prove a first Brudno's like relation between complexity and
entropy rates.

\begin{prop}
\label{GB}
Let $(\mathcal{X}, T, \nu)$ be an ergodic dynamical system.
Then,
\begin{equation}
\label{GB1}
g_\nu(T)=h_\nu^{KS}(T)\ .
\end{equation}
\end{prop}

\begin{proof}
Let $\mathcal{P}$ be a finite measurable partition such that $\nu_{\mathcal{P}}$ is computable.
By Levin's relation \eqref{thm:levin2} and inequality (\ref{re:kol}), we have
\begin{eqnarray*}
-\log{c_1}+ C(\ii)&\leq&  -\log\mu(\ii)\leq K(\ii)+\log{c}\\
&\leq& C(\ii)+2
\log{n} + \log{c_2}\ ,
\end{eqnarray*}
for all $\ii\in \Omega^{(n)}_k$,  where $c_1>0$ and $c_2>0$ are constant numbers.
Then, we may replace \eqref{eq:classical gacs} by
$$
g_\nu(T,\mathcal{P}^{(n)})=\limsup_{n\to\infty} \frac{1}{n}\sum_{\ii\in\Omega_k^{(n)}}
\nu_\mathcal{P}(\ii) C(\ii)\ .
$$
On the other hand, Brudno's theorem (see \eqref{Brudnoeasy}) ensures us that, for any
$\epsilon>0$, there is an
integer number $N_\epsilon$ such that for any $\n\ni n\geq N_\epsilon$,
$$
\frac{1}{n}C(\ii)\geq h^{KS}_\nu(T)-\epsilon\ ,
$$
for almost all sequences $\ii\in\Omega_k$ with respect to the measure
$\nu_{\mathcal{P}}$.
Therefore, by Lemma \ref{thm:class},
\begin{eqnarray*}
h_\nu^{KS}(T)&\geq& g_\nu(T)\geq g_\nu(T, \mathcal{P})
=
\limsup_{n\to\infty} \frac{1}{n}\sum_{\ii\in\Omega^{(n)}_k} \nu_\mathcal{P}(\ii)
C(\ii)\\
&\geq&
\limsup_{n\to\infty} \frac{1}{n}\sum_{\ii\in\Omega^{(n)}_k}
\nu(\ii)(h^{KS}_\nu(T)-\epsilon)
\geq
h^{KS}_\nu(T)-\epsilon\ ,
\end{eqnarray*}
whence
$h_\nu^{KS}(T)= g_\nu(T)$.
\end{proof}

The relation~\eqref{GB1} is a Brudno-like relation; we would like now to derive from it the full Brudno's result as stated in~\eqref{fullBrudno}. In order to prove it we first consider the case of an ergodic source and show a relation as in \eqref{Brudnoeasy}.

\begin{prop}
\label{propaux}
Let $(\Omega_k,\theta, \pi)$ be an ergodic source where $\pi$ is a computable probability measure with KS-entropy $h_\pi^{KS}(\theta)$. Then, for almost all $\bi\in\Omega_k$ with respect to $\pi$,
$$
\lim_{n\to\infty} -\frac{\log\mu(\ii)}{n}= h_\pi^{KS}(\theta)\ ,
$$
where $\ii$ is the starting segment of $\bi$ of length $n$.
\end{prop}

\begin{proof}
We start by proving that the inequality
$$
\lim_{n\to\infty} -\frac{\log\mu(\ii)}{n}\,\leq\, h_\pi^{KS}(\theta)\ ,
$$
holds almost everywhere with respect to $\pi$.
Let us consider the probability measure on $\Omega^{(n)}_k$:
\begin{equation}\label{inek1}
f(\ii)=\frac{1}{\sum_{n=1}^\infty n^{-2}}\frac{1}{n^2}\pi(\ii)\ ,
\end{equation}
where $\sum_{n=1}^\infty n^{-2}<+\infty$.
Since $\pi$ is computable, the same is true of $f$.
Then, by the universality of the semi-measure $\mu$, there exists a constant number $c_f>0$ such that
$$
\frac{c_f}{\sum_{n=1}^\infty n^{-2}} \frac{1}{n^2}\pi(\ii)=c_f\,f(\ii)\leq
\mu(\ii)\ .
$$
Then, the Shannon-Mc Millan-Breiman theorem (see~\eqref{thm:shannon-macmilla-breiman}) yields
\begin{equation}
\label{inek2}
\limsup_{n\to\infty} -\frac{\log\mu(\ii)}{n}\leq
\limsup_{n\to\infty} -\frac{\log\pi(\ii)}{n}=
h_\pi^{KS}(\theta),\quad \pi-a.e \ .
\end{equation}
The inequality
$$
\liminf_{n\to\infty} -\frac{\log\mu(\ii)}{n}\,\geq\, h_\pi^{KS}(\theta),\quad \pi-a.e\ ,
$$
follows from a counting argument as in the original proof
of Brudno's theorem~\cite{Brudno} that we shortly sketch.
From the Asymptotic Equipartition Property (AEP) and the Shannon-Mc
Millan-Breiman theorem ~\cite{billingsley}, we know
that, given the set
$$
A_\epsilon^{(n)}=\{\ii\in\Omega_2^{(n)}|
2^{-n(h_\pi+\epsilon)}\leq \pi(\ii)\leq
2^{-n(h_\pi-\epsilon)}\}\ ,
$$
with $h_\pi=h_\pi^{KS}(\theta)$, by choosing $n$ large enough one can make this set $\epsilon$-typical in the sense that
\begin{equation}
\label{typicalset}
Prob(A_\epsilon^{(n)})\geq 1-\epsilon\ ,\qquad (1-\epsilon)
2^{n(h_\pi-\epsilon)}<\#(A_\epsilon^{(n)})<2^{-n(h_\pi+\epsilon)}\ .
\end{equation}
On the other hand, by~\eqref{thm:levin2}, there exists a constant $c'>0$ such that
$$
-\log\mu(\ii)< K(\ii)+c'\qquad\forall\ \ii\in\Omega^*_2\ .
$$
Using~\eqref{compression}, one gets
$-\log\mu(\ii)\leq  C(\ii) + 2\log n + c'+c''$. Therefore,
\begin{eqnarray}
\nonumber
\#\{\ii:\, -\log\mu(\ii)< c   \} &\leq& \#\{\ii:\, C(\ii) + 2\log n + c'+c''<c   \} \\
\nonumber
&\leq& 2^{c-2\log n - c'-c''}\quad\hbox{whence}\\
\label{compreseeing-semimeasure}
\#\{\ii: \mu(\ii)&\geq& 2^{-c'}\}\leq 2^{c'+\alpha(n)}\ ,
\end{eqnarray}
where $\alpha(n)=-2\log n - c'-c''$.
Let us now consider the following subset of $A_\epsilon^{(n)}$,
\begin{equation}
\label{AEPT}
\hat{A}_\epsilon^{(n)}=\{\ii \in A_\epsilon^{(n)}\ :\ \mu(\ii)\geq
2^{-n(h_\pi-2\epsilon)}\}\ .
\end{equation}
Its measure can be bounded by
\begin{align}
\pi(\hat{A}_\epsilon^{(n)})
&\leq  \#(\hat{A}_\epsilon^{(n)})
\cdot \max_{\ion} \pi(\ii) \nonumber\\
 & \leq 2^{n(h_\pi-2\epsilon)+\alpha(n)+1}\cdot
2^{-n(h_\pi-\epsilon)}=2^{-n\epsilon+\alpha(n)+1}\ . \label{auxBrudno1}
\end{align}
As for the strings $\ii\notin A_\epsilon^{(n)}$
such that $\mu(\ii)\geq 2^{-n(h_\pi-2\epsilon)}$, let
$$
\tilde{A}_\epsilon^{(n)}=\{\ii\ :\ \mu(\bi^{(n)})\geq
2^{-n(h_\pi-2\epsilon)},\quad \ii\in (A_\epsilon^{(n)})^c\}\ ,
$$
where $(A_\epsilon^{(n)})^c=\Omega_2 \backslash
A_\epsilon^{(n)}$. Since $\tilde{A}_\epsilon^{(n)}\subseteq(A_\epsilon^{(n)})^c$ and
$A_\epsilon^{(n)}$ is typical (see~\eqref{typicalset}), $\pi((A_\epsilon^{(n)})^c)$ can be made arbitrarily small by choosing $n$ large enough and, because of~\eqref{auxBrudno1}, the same is true of
$\pi(\hat{A}_\epsilon^{(n)}\cup
\tilde{A}_\epsilon^{(n)})$.
Therefore, if $\ii\in\Omega^*_2$ is $\epsilon$-typical, one has
$$
\liminf_{n\to\infty} -\frac{\log\mu(\ii)}{n}\geq
h_\pi^{KS}(\theta)-2 \epsilon,\quad \pi-a.e\	 .
$$
\end{proof}

The previous results provide another proof of Brudno's Theorem with respect to the original one in~\cite{}.
\medskip

\begin{thm}
Let $(\mathcal{X}, T, \nu)$ be an ergodic dynamical system  with KS-entropy $h_\nu^{KS}(T)$.  Then,
$$
c(x)=\sup_{\nu_\mathcal{P}}c_\mathcal{P}(x)= h_\nu^{KS}(T),\quad \nu-a.e\ ,
$$
where the $\sup$ is taken over all $\mathcal{P}$  such that the probability measure $\nu_\mathcal{P}$ is computable probability measure.
\end{thm}

\begin{proof}
Let $\mathcal{P}$ be a finite measurable partition of $\mathcal{X}$, with $k$ atoms,
such that $\nu_\mathcal{P}$ is a computable probability measure. From~\eqref{inek2},
$$
\limsup_{n\to\infty} -\frac{\log\mu(\ii_\mathcal{P}(x))}{n}= h_{\nu_\mathcal{P}}^{KS}(\theta, \mathcal{P})\leq h_\nu^{KS}(\theta),\quad \nu-a.e\ ,
$$
where $\ii_\mathcal{P}(x)\in\Omega^{(n)}_k$.
Therefore, using~\eqref{thm:levin2}, we derive
$$
c_{\mathcal{P}}(x)\leq h_\nu^{KS}(\theta),\quad \nu-a.e\ ,
$$
whence, by taking the $\sup$ over all partitions such that $\nu_{\mathcal{P}}$ is semi-computable, we have
$$
c(x):=\sup_\mathcal{P}c_{\mathcal{P}}(x)\leq h_\nu^{KS}(\theta)\ ,\quad \nu-a.e\ .
$$
The proof of the inequality $c(x)\geq  h_\nu^{KS}(\theta)$ again follows from
the counting argument as in the proof of Proposition~\ref{propaux}.
\end{proof}

\section{Brudno's Relation: Quantum Spin Chains}

By means of Gacs complexity $\oh(\rho)$ and the quantum Shannon-Mac Millan
theorem~\cite{BSM}, we are now able to extend Brudno's result
to the shift over ergodic quantum spin chains providing a quantum version of the $\nu$-a.e. classical condition which is missing in~\cite{FSA}.
\medskip

We start by showing that there exists an effective, algorithmic procedure to construct semi-computable states, that are also faithful, namely, such that their local finite dimensional restrictions $\rho^{(n)}$ have no zero eigenvalues. Indeed, such a restriction is necessary for the proof of the main result in Theorem \ref{thm:brudnoq}.

Let us consider a sequence of semi-computable semi-density matrices $\rho^{(n)}$; namely, for each fixed $n$ there exists a sequence of quasi-increasing elementary semi-density matrices $\rho^{(n)}_m$  such that $\lim_m\rho^{(n)}_m=\rho^{(n)}$ in trace-norm.
\medskip

\begin{defn}
A faithful state $\omega$ on $\mathcal{M}$ is called semi-computable
(computable) if the associated sequence of local density matrices
$\rho^{(n)}$ on $M_{[-n, n]}$ with $rank(\rho^{(n)}_m)=2^{2n+1}$,
is such that the  semi-computable (computable)
semi-density matrices and the function $(m, n)\to a_{n
m}$ from $\n\times\n \to \n$ is  computable.
\end{defn}

\begin{rem}
\label{semcomprem}
In general, given a sequence of unitary operator $U^{(n)}$ with
rational entries and a given quasi-increasing sequence of semi-computable semi-density matrix $\rho^{(n)}$, the sequence of semi-computable semi-density matries $U^{(n)}\,\rho^{(n)}\,(U^{(n)})^\dag$ need not be quasi-increasing.
\qed
\end{rem}

The previous observation motivates the following preliminary auxiliary result.

\begin{lem}
\label{semi-unive}
Let $(\mathcal{M}, \Theta, \omega)$ be a quantum spin chain with
$\omega$ a semi-computable faithful state, whose restrictions to the local
algebras $M_{[0,n-1]}$ correspond to density matrices
of $\rho^{(n)}$ of full rank $2^n$. Let
\begin{equation}
\label{musp}
\hat{\mu}^{(n)}=\sum_{\ii\in\Omega_2^{(n)}}\,\mu_{\ii}\,\vert\mu_{\ii}\rangle\langle\mu_{\ii}|\ ,
\end{equation}
be the spectral representation of the restriction of the universal semi-density matrix $\hat \mu$ in~\eqref{usdm} to the local Hilbert space $\mathbb{H}_{[0,n-1]}$ and let
\begin{equation}
\label{rosp}
\rho^{(n)}=\sum_{\ii\in\Omega_2^{(n)}} r_{\ii}|r_{\ii}><r_{\ii}|\ ,
\end{equation}
be the spectral representation of $\rho^{(n)}$. Finally, define $U^{(n)}$ as the unitary operator transforming the spectral support of $\rho^{(n)}$ onto that of $\hat\mu^{(n)}$:
$U^{(n)}|\mu_{\ii}>=|r_{\ii}>$.
Then,
$$
\limsup_{n\to\infty }\frac{1}{n}\log{\rm Tr}(\sigma^{(n)} \hat{\mu})=\limsup_{n\to\infty }\frac{1}{n}\log{\rm Tr}(\sigma^{(n)}\, U^{(n)} \hat{\mu}\, (U^{(n)})^\dag)\ ,
$$
for any semi-computable density matrix $\sigma^{(n)}\in M_{[0, n-1]}$.
\end{lem}
\medskip

\begin{rem}
\label{rem2}
Notice that both $\hat{\mu}^{(n)}$ and $\rho^{(n)}$ are $2^n\times 2^n$ matrices with $2^n$
eigenvalues that can be listed in decreasing order and then associated to the binary encodings $\ii\in\Omega^{(n)}_2$ of their labels in the list.
Since the universal semi-density matrix is full rank, the $\rho^n$'s  must be full rank, too, whence the demand of faithfulness of the semi-computable semi-density matrices $\rho^{(n)}$ .
\qed
\end{rem}
\bigskip

\begin{proof}
By \cite{FSA}, the local restrictions $\hat{\mu}^{(n)}=P^{(n)} \hat{\mu} P^{(n)}$, where
$P^{(n)}$ projects from $\mathbb{H}$ onto $\mathbb{H}_{[0,n-1]}$ for each
$n\in\n$, are universal semi-density matrices in $M_{[0, n-1]}$.

Since  $\rho^{(n)}$ and $\hat{\mu}^{(n)}$ are semi-computable
density matrices, there exist computable, quasi-increasing elementary matrices $\{\rho^{(n)}_m\}_{m\in\n}$ and $\{\hat{\mu}^{(n)}_m\}_{m\in\n}$
such that $\lim_m\rho^{(n)}_m=\rho^{(n)}$, respectively
$\lim_m\hat{\mu}^{(n)}_m=\hat{\mu}^{(n)}$, in trace-norm.
Moreover, the global spin-chain state
$\omega$ is assumed to be faithful, so that the ranks of $\rho^{(n)}_m$
and $\hat{\mu}^{(n)}_m$ can be taken equal to $2^n$, for each $n$.
Therefore, the unitary operators $U^{(n)}_m$ sending the spectral support of
the elementary matrix $\rho^{(n)}_m$ into that of the elementary matrix
$\hat{\mu}^{(n)}_m$ have rational entries.

On the other hand, $\hat{\mu}$ is a semi-computable semi-density matrix; then,
there exits a quasi-increasing (see Definition~\ref{semicomp}) sequence of elementary semi-density matrices $\hat{\mu}^{(k)}$ such that $\lim_k\hat{\mu}^{(k)}=\hat{\mu}$ in trace-norm. Therefore, for fixed $n$ and $m$, when $k\to\infty$, the elementary semi-density matrices $(U_m^{(n)})^\dag\hat{\mu}^{(k)}\, U^{(n)}_m$ form a quasi-increasing sequence converging to $(U^{(n)}_m)^\dag \hat{\mu}\, U_m^{(n)}$ which is thus also a semi-computable semi-density matrix. Let us consider the following  operator
\begin{equation}
\hat{K}_m=\sum_{n\geq2} \frac{1}{n \log^2 n} (U_m^{(n)})^\dag \hat{\mu} U^{(n)}_m\ ,
\end{equation}
which, by construction, is a semi-computable semi-density matrix for each $m\in\n$. Therefore, there exists a
constant $c_m>0$ such that
\begin{equation}\label{in:unitary}
 c_m \frac{(U_m^{(n)})^\dag \hat{\mu}\, U_m^{(n)}}{n \log^2 n} \leq c_m \hat{K}_m \leq \hat{\mu}\ ,
 \end{equation}
 and hence
\begin{equation}
\hat{\mu} \leq  \frac{n \log^2 n}{c_m}\, U_m^{(n)} \hat{\mu} (U^{(n)}_m)^\dag\ .
\end{equation}
Let $\sigma^{(n)}$ be a semi-density matrix on $M_{[-n, n]}$. Then,
\begin{equation}
\limsup_{n\to\infty }\frac{1}{n}\log{\rm Tr}(\sigma^{(n)} \hat{\mu})\leq
\limsup_{n\to\infty }\frac{1}{n}\log {\rm Tr}\Big(\sigma^{(n)} U_m^{(n)} \hat{\mu}\ (
U_m^{(n)})^\dag\Big)\ .
\end{equation}
On the other hand,
\begin{eqnarray}
\nonumber
&&
\left|{\rm Tr}\Big(\sigma^{(n)} U_m^{(n)} \hat{\mu}\ (U_m^{(n)})^\dag-
\sigma^{(n)} U^{(n)} \hat{\mu}\ (U^{(n)})^\dag\Big) \right|\ \leq\\
\nonumber
&&
\leq
\|\sigma^{(n)}\|\,\|U_m^{(n)} \hat{\mu} (U_m^{(n)})^\dag-  U^{(n)} \hat{\mu}\ (U^{(n)})^\dag \|_1\\
\nonumber
&&
\leq\| U_m^{(n)} \hat{\mu}\ (U_m^{(n)})^\dag-  U^{(n)} \hat{\mu}\ (U^{(n)})^\dag\|_1\\
\nonumber
&&
\leq\left\| U_m^{(n)} \hat{\mu}\ \Big((U_m^{(n)})^\dag\,-\,(U^{(n)})^\dag\Big)\right\|_1
\,+\,\left\|\Big(U_m^{(n)}-U^{(n)}\Big) \hat{\mu}\ (U^{(n)})^\dag\right\|_1\\
\nonumber
&&
\leq
\|U_m^{(n)}\, \hat{\mu}\|\,\|(U_m^{(n)})^\dag -
(U^{(n)})^\dag\|_1\, +\, \|\hat{\mu}\, (U^{(n)})^\dag\|\,\|U_m^{(n)}- U^{(n)})\|_1\\
\label{in:unitarylimit}
&&
\leq\,2\,\|U_m^{(n)} - U^{(n)}\|_1\ .
\end{eqnarray}
Thefore, since $\lim_m U^{(n)}_m=U^{(n)}$ in trace-norm,
$$
\limsup_{n\to\infty }\frac{1}{n}\log{\rm Tr}(\sigma^{(n)} \hat{\mu})\leq \limsup_{n\to\infty }\frac{1}{n}\log {\rm Tr}(\sigma^{(n)} U^{(n)} \hat{\mu}\ (U^{(n)})^\dag )\ .
$$
A similar argument shows that the inequality can be inverted.
Indeed, from the relation \eqref{in:unitary}, we have
$$
\limsup_{n\to\infty }\frac{1}{n}\log {\rm Tr}\Big(\sigma^{(n)} U_m^{(n)} \hat{\mu}\ (
U_m^{(n)})^\dag\Big)\ \leq \limsup_{n\to\infty }\frac{1}{n}\log{\rm Tr}(\sigma^{(n)} \hat{\mu})\ .
$$
Finally, using~\eqref{in:unitarylimit} and taking the limit when $m\to\infty$ yield
$$
\limsup_{n\to\infty }\frac{1}{n}\log {\rm Tr}\Big(\sigma^{(n)} U^{(n)} \hat{\mu}\ (
U^{(n)})^\dag\Big)\ \leq \limsup_{n\to\infty }\frac{1}{n}\log{\rm Tr}(\sigma^{(n)} \hat{\mu})\ .
$$

\end{proof}

In~\cite{FSA}, a Brudno type relation was established between the von Neumann entropy rate and the Gacs complexity rate along a sequence of local restrictions of any shift-invariant state on a quantum spin chain. In order to fully extend the classical Brudno relation to the quantum setting, one ought to introduce a quantum analog of the almost everywhere condition in~\eqref{fullBrudno}. A similar problem was encountered in~\cite{BKMSSS}, where use was made of the notion of $\epsilon$-typicality (condition~\eqref{QSMM} in the theorem below) of minimal projectors with respect to a given state.
The latter condition was an essential ingredient in establishing a quantum version~\cite{BSM} of the classical Shannon-Mc Millan-Breiman theorem that we now briefly introduce.
\medskip

\begin{thm} (Quantum Shannon-McMillan Theorem).
Let $\omega$ be an ergodic state on $\mathcal{M}$ with von Neumann entropy rate $s(\omega)$. Then, for all $\epsilon> 0$ there is an $N_\epsilon$ such that for all $n\geq N_\epsilon$ there exists a projector $p^{(n)}(\epsilon)\in M_{[-n,n]}$ such that
\begin{itemize}
\item
it is $\epsilon$-typical, namely
\begin{equation}
\label{QSMM}
{\rm Tr}\Big(\rho^{(n)}\,p^{(n)}(\epsilon)\Big)) \geq 1 -\epsilon\ ;
\end{equation}
\item
for all minimal projectors $p\in M_{[-n,n]}$ with $p\leq p^{(n)}(\epsilon)$ one has
\begin{eqnarray}
\label{QSMMa}
&&
{\rm e}^{-n(s(\omega)+\epsilon)} < {\rm Tr}\Big(\rho^{(n)}\,p\Big) < {\rm e}^{-n(s(\omega)-\epsilon)}\\ 
\label{QSMMb}
&&
{\rm e}^{n(s(\omega)-\epsilon)} < {\rm Tr}(p^{(n)}(\epsilon)) < {\rm e}^{n(s(\omega)+\epsilon)}\ .
\end{eqnarray}
\end{itemize}
\end{thm}

Practically speaking, typical projections with respect to a given state over a quantum spin chain project onto subspaces of high probability relative to that state. If the state is ergodic with respect to the shift along the chain, then, for sufficiently large $n$, there exist typical projections in each local sub-algebras $M_{[-n,n]}$ projecting onto subspaces of dimension close to the exponential of $n$ times the von Neumann entropy rate of the chain. Moreover, any projector onto a state vector in such subspaces has a mean value with respect to the ergodic state which is  close to the inverse of the probability of the subspace.
With these tools at disposal, we can now prove the main result in the quantum case which generalizes the result in~\cite{FSA}.

\begin{thm}
\label{thm:brudnoq}
Let $(\mathcal{M}, \Theta, \omega)$, be an ergodic quantum
spin-chain with right shift dynamics $\Theta$ and faithful
semi-computable state $\omega$ with restrictions to local algebras $M_{[0,n-1]}$ corresponding to density matrices $\rho^{(n)}$ of rank $2^n$. Then, for any $\epsilon
>0$, there exists a sequence of projections $p^{(n)}(\epsilon)\in M_{[0, n-1]}$
and a number $N_\epsilon\in \n$ such that for all $n\geq
N_\epsilon$, we have that
\begin{enumerate}
\item
$p^{(n)}(\epsilon)$ projects onto a high-probability subspace:
$$
\omega(p^{(n)}(\epsilon))={\rm Tr}\Big(\rho^{(n)}\,p^{(n)}(\epsilon)\Big)>1-\epsilon - 2^{-n\epsilon+\alpha_n}\ ;
$$
\item
the subspace dimension is controlled by the von Neumann entropy rate:
$$
(1-\epsilon-2^{-n \epsilon+\alpha_n})2^{n(s(\omega)-\epsilon)}< {\rm Tr}_n(p^{(n)}(\epsilon))  < 2^{n(s(\omega)+\epsilon)}\ ,
$$
where $\lim_{n\to \infty} \frac{\alpha_n}{n} =0$;
\item
for all minimal projections $0 \neq p^{(n)}\in M_{[0, n-1]}$ dominated by $p^{(n)}(\epsilon)$, $p^{(n)}\leq p^{(n)}(\epsilon)$, we have
$$
2^{-n(s(\omega)+\epsilon)}\leq \omega(p^{(n)})\leq    2^{-n(s(\omega)-\epsilon)}\ ;
$$
\item
while their Gacs complexities obey
$$
\lim_{n\to\infty} -\frac{1}{n}\log{\rm Tr}\Big(\hat{\mu}\,p^{(n)}\Big)=s(\omega)\ .
$$
\end{enumerate}
\end{thm}
\medskip

\begin{proof}
Let the state and universal semi-density matrix restrictions to $M_{[0,n-1]}$,
$\rho^{(n)}$ and $\hat{\mu}^{(n)}$, be spectrally decomposed as in~\eqref{rosp}
and~\eqref{musp}.
Let then introduce the subsets
\begin{eqnarray}
\label{eq1}
\hskip -1cm
\Omega^{(n)}_2&\supseteq&
A^{(n)}_\epsilon=\Big\{\ii\in\Omega^{(n)}_2 : 2^{-n(s(\omega)+\epsilon)}
\leq r_{\ii}\leq 2^{-n(s(\omega)-\epsilon)}\Big\}\\
\label{eq2}
\hskip-1cm
\Omega^{(n)}_2&\supseteq&
B^{(n)}_\epsilon=\Big\{\ii\in\Omega^{(n)}_2 : \mu_{\ii}<
2^{-n(s(\omega)-2\epsilon)}\Big\}\ .
\end{eqnarray}

Using~\eqref{compreseeing-semimeasure}, the cardinality of the complement $(B_\epsilon^{(n)})^c$ of the latter subset is bounded from below by
\begin{equation}
\label{auxeq1}
\hbox{card}\Big((B_\epsilon^{(n)})^c\Big) \leq 2^{n(s(\omega)-2 \epsilon)+\alpha_n}\ ,
\end{equation}
where $\alpha_n>0 $ is a constant such that
$\lim_{n\to \infty} \frac{\alpha_n}{n} =0$.

On the other hand, from the quantum Shanonn-MacMillan
Theorem ~\cite{BSM}, one knows that
\begin{equation}
\label{auxeq2}
\sum_{\ii\in\A^{(n)}_\epsilon} r_{\ii}>1-\epsilon\ .
\end{equation}

Consider now the sequence of spectral
projections $p^{(n)}(\epsilon)\in M_{[0, n-1]}$
$$
p^{(n)}(\epsilon)=\sum_{\ii\in A^{(n)}_\epsilon\cap
B^{(n)}_\epsilon}\vert r_{\ii}\rangle\langle r_{\ii}\vert\ .
$$

Let $M_{[0, n-1]}\ni p^{(n)}\leq p^{(n)}(\epsilon)$
be a minimal projection onto a vector state, $p^{(n)}=\vert\psi^{(n)}\rangle\langle\psi^{(n)}\vert$, with
\begin{equation}
\label{re:normalize}
\vert\psi^{(n)}\rangle=\sum_{\ii\in A^{(n)}_\epsilon\cap
B^{(n)}_\epsilon} c_{\ii} \vert r_{\ii}\rangle\ ,\quad
\sum_{\ii\in A^{(n)}_\epsilon\cap B^{(n)}_\epsilon} |c_{\ii}|^2=1\ .
\end{equation}

\textbf{Proof of $1$}\quad Writing $A^{(n)}_\epsilon$ as the union of two disjoint subsets $$
A^{(n)}_\epsilon=\left(A^{(n)}_\epsilon\cap
B^{(n)}_\epsilon\right)\bigcup\Big(A^{(n)}_\epsilon\backslash B^{(n)}_\epsilon\Big)\ ,
$$
using~\eqref{auxeq1} and~\eqref{auxeq2} one estimates
\begin{eqnarray*}
{\rm Tr}\Big(\rho^{(n)} p^{(n)}(\epsilon)\Big)
&=&\sum_{\ii\in A^{(n)}_\epsilon\cap B^{(n)}_\epsilon} r_{\ii}\\
&=&
\sum_{\ii\in A^{(n)}_\epsilon} r_{\ii}-\sum_{\ii\in
A^{(n)}_\epsilon\backslash B^{(n)}_\epsilon} r_{\ii}\\
&\geq& 1-\epsilon - \sum_{\ii\in A^{(n)}_\epsilon\backslash
B^{(n)}_\epsilon} 2^{-n(s(\omega)-\epsilon)}\\
&\geq&
 1-\epsilon - 2^{-n(s(\omega)-\epsilon)}\hbox{card}\Big\{(B^{(n)}_\epsilon)^c\Big\}\\
 &\geq&
 1-\epsilon - 2^{-n(s(\omega)-\epsilon)} 2^{n(s(\omega)-2\epsilon)+\alpha_n}\\
 &\geq&
 1-\epsilon - 2^{-n\epsilon+\alpha_n}\ ,
\end{eqnarray*}
where, in the latter quantity $2^{-n \epsilon +
\alpha_n}$ can be made negligibly small by increasing $n$.
\medskip

\noindent
\textbf{Proof of $2$:}\quad Using~\eqref{eq1}, one gets
$$
{\rm Tr}\rho=1\geq\sum_{\ii\in A^{(n)}_\epsilon}r_{\ii}\geq 2^{-n(s(\omega)+\epsilon)}\,\hbox{card}\Big\{\ii\in A^{(n)}_\epsilon\Big\}\ .
$$
Furthermore, using~\eqref{auxeq1} one estimates
\begin{eqnarray*}
1-\epsilon& \leq&
\sum_{\ii\in A^{(n)}_\epsilon } r_{\ii}\leq \sum_{\ii\in A^{(n)}_\epsilon } 2^{-n(s(\omega)-\epsilon)}\\
&\leq& \hbox{card}\Big\{\ii\in A^{(n)}_\epsilon\Big\} 2^{-n(s(\omega)-\epsilon)}\ .
\end{eqnarray*}
Then, by construction,
\begin{eqnarray*}
{\rm Tr}(p^{(n)}(\epsilon))
&=& \hbox{card}\Big\{\ii\in A^{(n)}_\epsilon\cap
B^{(n)}_\epsilon\Big\}\leq\hbox{card}\Big\{\ii\in A^{(n)}_\epsilon\Big\}\\
&\leq&
2^{n(s(\omega)+\epsilon)}\\
{\rm Tr}(p^{(n)}(\epsilon))&=&
\hbox{card}\Big\{\ii\in A^{(n)}_\epsilon\cap
B^{(n)}_\epsilon\Big\}\\
&=&
\hbox{card}\Big\{\ii\in A^{(n)}_\epsilon\Big\}\,-\,\hbox{card}\Big\{\ii\in
A^{(n)}_\epsilon\backslash B^{(n)}_\epsilon\Big\}\\
&\geq&
(1-\epsilon)2^{n(s(\omega)-\epsilon)}- 2^{n(s(\omega)-2\epsilon) + \alpha_n}\\
&\geq& ((1-\epsilon)-2^{-n \epsilon+\alpha_n})\,2^{n(s(\omega)-\epsilon)}\ .
\end{eqnarray*}

\noindent
\textbf{ Proof of $3$:}\quad
Using~\eqref{eq1} and the normalization of $\vert\psi^{(n)}\rangle$, for any minimal projection $p^{(n)}\leq p^{(n)}(\epsilon)$, one estimates
$$
\omega(p^{(n)}) = \sum_{\ii\in A^{(n)}_\epsilon\cap B^{(n)}_\epsilon}
r_{\ii} |c_{\ii}|^2
\leq 2^{-n(s(\omega)-\epsilon)}\ ,
$$
while from~\eqref{eq1} and~\eqref{re:normalize}, it follows that
$$
\omega(p^{(n)}) = \sum_{\ii\in A^{(n)}_\epsilon\cap
B^{(n)}_\epsilon} r_{\ii}\,\vert c_{\ii}\vert^2\geq
 2^{-n(s(\omega)+\epsilon)}\ .
$$

\noindent
\textbf{Proof of $4$:}\quad Since the quantum state is assumed semi-computable, such is the semi-density matrix
$$
\eta=\sum_{n=2}^\infty\delta(n) \rho^{(n)}\ ,
$$
with coefficients $\delta(n)\geq 0$ that ensures convergence in trace-norm.
Therefore, there exists a constant $c>0$ such that $c\delta(n)\,
\rho^{(n)}\leq c\,\eta\leq \hat{\mu}$, for all $n\in\n$, whence
\begin{eqnarray*}
\limsup_{n\to\infty}-\frac{1}{n}\log{\rm Tr}(\hat{\mu}p^{(n)})
 &\leq&
\limsup_{n\to\infty}-\frac{1}{n}\log{\rm Tr}(\rho^{(n)}
p^{(n)})\\
 &\leq&
\limsup_{n\to\infty}-\frac{1}{n}\log\omega(p^{(n)})\\
& \leq&
\limsup_{n\to\infty}-\frac{1}{n}\log  \left((1-2^{-n\epsilon}) 2^{-n(s(\omega)+\epsilon)+\alpha_n} \right)  \\
&\leq&
 s(\omega)+\epsilon +\alpha_n.
\end{eqnarray*}
Thus,
$$ \limsup_{n\to\infty} -\frac{{\rm Tr}(\hat{\mu} p^{(n)})}{n}\leq  s(\omega)+\epsilon\ . $$
Since the rank of $\rho^{(n)}$ is $2^n$,
the operator
$$
\hat{T}^{(n)}=\sum_{\ii\in\Omega_2^{(n)}} \mu_{\ii}\,\vert r_{\ii}
\rangle\langle r_{\ii}\vert\ ,
$$
is a semi-computable semi-density matrix (see Remark~\ref{rem2}).

Therefore,  there exists $c_n>0$ such that $c_n
\hat{T}^{(n)}\leq \hat{\mu}^{(n)}$.
Let $U^{(n)}$ be the unitary operator such that
$U^{(n)}\vert\mu_{\ii}\rangle=\vert r_{\ii}\rangle$ (see~\eqref{musp} and~\eqref{rosp}).
By Lemma \ref{semi-unive}, we have
\begin{eqnarray*}
&&
\liminf_{n\to\infty}-\frac{1}{n}\log{\rm Tr}\Big(\hat{\mu}\ p^{(n)}\Big)
\geq
\liminf_{n\to\infty}-\frac{1}{n}\log{\rm Tr}\Big( U^{(n)} \hat{\mu}^{(n)}\ (U^{(n)})^\dag p^{(n)}\Big)\\
&&\hskip 1cm
\geq
\liminf_{n\to\infty}-\frac{1}{n}\log{\rm Tr}\Big(\hat{T}^{(n)}\ p^{(n)}\Big)\\
&&\hskip 1cm
\geq
\liminf_{n\to\infty}-\frac{1}{n}\log\Big(\sum_{\ii\in \Omega^{(n)}_2}
\mu_{\ii}\,\langle r_{\ii}\vert p^{(n)}\vert r_{\ii}\rangle\Big)\\
&&\hskip 1cm
=
\liminf_{n\to\infty}-\frac{1}{n}\log\Big(\sum_{\ii\in A^{(n)}_\epsilon \cap B^{(n)}_\epsilon}
\mu_{\ii}\,\langle r_{\ii}\vert p^{(n)}\vert r_{\ii}\rangle\Big)\\
&& \hskip 1cm
\geq
\liminf_{n\to\infty}-\frac{1}{n}\log\Big(2^{-n(s(\omega)-2\epsilon)}
\sum_{\ii\in A^{(n)}_\epsilon \cap B^{(n)}_\epsilon}\,\langle r_{\ii}\vert p^{(n)}\vert r_{\ii}\rangle\Big)\\
&&\hskip 1cm
= s(\omega)-2\epsilon\ ,
\end{eqnarray*}
where the two equalities follow since $p^{(n)}$ is a projection such that
\begin{eqnarray*}
&&
\langle r_{\ii}\vert p_n\vert r_{\ii}\rangle=0\qquad \forall \ii\in \Big(A^{(n)}_\epsilon \cap B^{(n)}_\epsilon\Big)^c\\
&&
\sum_{\ii\in A^{(n)}_\epsilon \cap B^{(n)}_\epsilon}\,\langle r_{\ii}\vert p^{(n)}\vert r_{\ii}\rangle=\sum_{\ii\in \Omega^{(n)}_2}\,\langle r_{\ii}\vert p^{(n)}\vert r_{\ii}\rangle={\rm Tr}(p^{(n)})=1\ ,
\end{eqnarray*}
while the last inequality can be derived from
$$
\mu_{\ii}\leq 2^{-n(s(\omega)-2\epsilon)}\qquad \forall \ii\in A^{(n)}_\epsilon \cap B^{(n)}_\epsilon\ ,
$$
that follows from the definition of the set $B^{(n)}_\epsilon$ in~\eqref{eq2}.
\end{proof}

\section{Conclusions}

We have applied the quantum complexity introduced by P. Gacs in~\cite{Gacs} in two different scenarios. The first one concerns its use in evaluating the complexity of the trajectories of classical ergodic dynamical systems: in such a case, we showed that the Gacs complexity rate of almost every trajectory equals the Kolmogorov-Sinai dynamical entropy exactly as Brudno's theorem does for the Kolmogorov complexity. 
The second scenario consists of a quantum spin chain endowed with a translational invariant ergodic state for which we proved a full quantum Brudno's relation in that the equality between the Gacs complexity rate and the chain entropy density, already shown 
in~\cite{FSA}, holds on a dense set of vector states. This last condition is the quantum counterpart of the almost everywhere condition in the classical formulation of Brudno's theorem.

\bigskip
\bigskip
\noindent
{\bf Acknowledgement}\quad Samad Khabbazi Oskouei is happy to acknowledge the support of the STEP programme of the Abdus Salam ICTP of Trieste.

\end{document}